\documentclass{article}

\usepackage{braket}
\usepackage{mathtools}
\usepackage[overload]{empheq} 
\usepackage{bm}
\usepackage{enumitem}
\usepackage{amssymb}
\usepackage{booktabs}
\usepackage{algorithm}
\usepackage{algpseudocode}
\usepackage{amsthm}
\usepackage{xcolor}

\usepackage[margin=2cm]{geometry}
\usepackage{newtxtext}
\usepackage{newtxmath}
\usepackage{setspace}
\usepackage[absolute]{textpos}

\renewcommand{\qed}{\flushright$\blacksquare$}

\algnewcommand\algorithmicinput{\textbf{Input:}}
\algnewcommand\INPUT{\item[\algorithmicinput]}

\algnewcommand\algorithmicoutput{\textbf{Output:}}
\algnewcommand\OUTPUT{\item[\algorithmicoutput]}

\let\oldReturn\Return
\renewcommand{\Return}{\State\oldReturn}

\makeatletter
\def\thmheadbrackets#1#2#3{%
	\thmname{#1}\thmnumber{\@ifnotempty{#1}{ }\@upn{#2}}%
	\thmnote{ {\the\thm@notefont[#3]}}}
\makeatother

\newtheoremstyle{brakets}
{}
{}
{\itshape}
{}
{\bfseries}
{.}
{ }
{\thmheadbrackets{#1}{#2}{#3}}

\newtheoremstyle{defbrakets}
{}
{}
{\normalfont}
{}
{\bfseries}
{.}
{ }
{\thmheadbrackets{#1}{#2}{#3}}

\newtheoremstyle{defproblem}
{}
{}
{\normalfont}
{}
{\bfseries}
{.}
{ }
{\thmheadbrackets{#1}{#2}{#3}}

\newtheorem{thm}{Theorem}
\newtheorem*{thm-non}{Theorem} 

\theoremstyle{definition}
\theoremstyle{defbrakets}
\newtheorem{defn}{Definition}
\newtheorem{rem}{Remark}
\newtheorem{plm}{Problem}
\newtheorem*{plm-non}{Problem}

\DeclareFontFamily{OT1}{pzc}{}
\DeclareFontShape{OT1}{pzc}{m}{it}{<-> s * [1.10] pzcmi7t}{}
\DeclareMathAlphabet{\mathpzc}{OT1}{pzc}{m}{it}

\makeatletter
\newcommand{\algrule}[1][.2pt]{\par\vskip.5\baselineskip\hrule height #1\par\vskip.5\baselineskip}
\makeatother

\begin{document}
\title{Supervisory Coordination of Robotic Fiber Positioners in Multi-Object Spectrographs}
\author{Matin Macktoobian, Denis Gillet, and Jean-Paul Kneib}
\date{Swiss Federal Institute of Technology in Lausanne (EPFL), 1015 Lausanne, Switzerland\\
Email: matin.macktoobian@epfl.ch\\[4mm]This work was financially supported by the Swiss National Science Foundation (SNF) grant number 20FL21\_185771 and the SLOAN ARC/EPFL agreement number SSP523.}
\maketitle
\begin{textblock}{14}(1.5,1)
	\noindent\textbf{\color{red}Published in ``Part of special issue:
		15th IFAC Symposium on Large Scale Complex Systems (LSS 2019)''\\ DOI: 10.1016/j.ifacol.2019.06.011}
\end{textblock}
\begin{abstract}                
In this paper, we solve the complete coordination problem of robotic fiber positioners using supervisory control theory. In particular, we model positioners and their behavioral specifications as discrete-event systems by the discretization of their motion spaces. We synthesize a coordination supervisor associated with a specific set of positioners. In particular, the coordination supervisor includes the solutions to the complete coordination problem of its corresponding positioners. Then, we use the backtracking forcibility technique of supervisory control theory to present an algorithm based on a completeness condition to solve the coordination problem similar to a reconfiguration problem. We illustrate the functionality of our method using an example.	
\end{abstract}
\textit{keywords}: coordination, robotic fiber positioners, discrete-event systems, supervisory control, spectroscopic surveys
\section{Introduction}
\doublespacing
Modern astronomy aims to study the evolution of the universe using cosmological spectroscopy \cite{mazets1982cosmic} surveys. Each survey is a map corresponding to a piece of the observable universe. For this purpose, many projects, e.g., DESI \cite{1aghamousa2016desi,2aghamousa2016desi}, MOONS \cite{cirasuolo2014moons}, etc., are carried out to develop telescopes equipped with spectrographs. Each spectrograph is connected to a set of fiber positioners. Each fiber positioner is assigned to a specific galaxy in the sky corresponding to a particular observation. Then, the spectral information of that galaxy is transferred to its spectrograph through the observing fiber positioner. To increase the throughput of each observation in view of the collected information, one takes positioner swarms into account. Optical fibers are mounted in a particular area of their hosting telescope called focal plane. The favorite hexagonal formation of positioners shall be dense enough to mount as many as possible positioners at the focal plane of the hosting telescope \cite{horler2018high,horler2018robotic}. Since the target assignment to an optical fiber is changed from one observation to another, a robotic positioner system is attached to each fiber to move it. Each robotic positioner is generally a rotational-rotational robot whose workspace overlaps with those of its neighbors. So, the collision avoidance, the trajectory planning, and the completeness associated with a set of positioners are challenging. In particular, a set of positioners is completely coordinated if all of its positioners point to their targets at the end of the coordination process corresponding to a specific observation. We seek the complete coordination under the assumption that all of the target positioners are assigned to some galaxies. 

The current solutions to the coordination problem of positioners lack general completeness. In particular, 
artificial potential fields \cite{macktoobian2013time,macktoobian2016time} are used to plan collision-free trajectories for positioners \cite{makarem2014collision}. The state of the art \cite{tao2018priority} combines a low-level artificial-potential-based navigator with a high-level state-machine-based decision maker to handle deadlocks and oscillations. So, the positioners with higher priorities are prioritized to be coordinated. However, the complete convergence of all positioners is not guaranteed. Supervisory control theory and discrete-event systems are promising candidates by which we attack the complete coordination problem. 

Discrete-event systems and supervisory control theory \cite{wonham2017supervisory} efficiently models and controls event-driven complex systems, respectively. In particular, automatic reconfiguration \cite{macktoob2017auto} is addressed using supervisory control theory. Each configuration of a discrete-event system exhibits a specific set of functionalities. The discrete-event system is reconfigured from one configuration to another by activating a controllable reconfiguration event. Then, a supervisor is synthesized based on the behavioral and the reconfiguration requirements of the discrete-event system. Then to execute the reconfiguration, the supervisor shall find a string of events from its current state to a target state at which the reconfiguration event is eligible to occur. Backtracking forcibility technique \cite{macktoobianautomatic} was developed to find the set of all forcible paths from the current state to the target state. 

The automatic reconfiguration of discrete-event systems and the coordination of robotic positioners are intrinsically similar in a systematic point of view. So in this paper, we seek a solution to the complete coordination problem of robotic positioners using supervisory control theory and reconfiguration of discrete-event systems. Our supervisory control approach models each robotic positioner as a rotational arm with only one DoF. This assumption is based on a trade-off according to which the state size of the overall system remains reasonably tractable in view of the supervisor computation and backtracking forcibility. On the other hand, this assumption may limit the number of available solutions to the problem; yet our coordination method is efficient enough to achieve the desired completeness. In particular, we model each positioner and its behavior as discrete-event systems. Then, we synthesize a supervisor, called \textit{coordination supervisor}, to control the overall behavior of the complex set of positioners. We define the plant model and the specifications such that the marked state of the synthesized supervisor represents the complete coordination of the system. Thus, the coordination problem is indeed reduced to finding forcible paths from the initial state of the supervisor to its target state. One notes that the cited process is similar to the solution checking of the reconfiguration problem briefly explained above. In other words, the complete coordination problem seeks the reachability of a particular state corresponding to a coordination supervisor. For this purpose, we propose an algorithm to realize the quoted completeness-checking process associated with a typical coordination problem. 

The remainder of the paper is organized as follows. A brief review of the supervisory control theory is represented in Section \ref{sec:background}. Section \ref{sec:supervision} defines the complete coordination problem in the language of a backtracking forcibility problem and proposes an algorithm to solve it. Section \ref{sec:example} includes an example to illustrate how our algorithm practically solves a complete coordination problem. Our concluding remarks are finally drawn in Section \ref{sec:conc}.
\section{Background}
\label{sec:background}
Supervisory control theory controls discrete-event systems modeled by the Ramadge-Wonham framework \cite{ramadge1987supervisory}. A discrete-event system is formally represented by a generator, say,
\begin{equation}
\textbf{G} = (Q,\Sigma,\delta,q_0,Q_m),
\end{equation}
\noindent where $\Sigma = \Sigma_{c} \dot{\cup} \Sigma_{u}$ is a finite alphabet of event labels, partitioned into the \textit{controllable} event labels and the \textit{uncontrollable} ones; $Q$ is the finite \textit{state set}; $\delta:Q \times \Sigma^{*} \rightarrow Q$ is the \textit{extended partial transition function}; $q_0$ is the \textit{initial state}; and $Q_{m} \subseteq Q$ is the subset of \textit{marked states}. The \textit{closed behavior} and the \textit{marked behavior} of \textbf{G} are the regular languages
\begin{subequations}
	\begin{empheq}[left={}]{align}
	L(\textbf{G}) &:= \{s\in \Sigma^{*}|\delta(q_0,s)!\},
	\label{eq:1}\\
	L_{m}(\textbf{G}) &:= \{s\in L(\textbf{G})|\delta(q_0,s) \in Q_{m}\}.
	\label{eq:2}
	\end{empheq}
\end{subequations}
\noindent Here $\delta(q_0,s)!$ means that $\delta(q_0,s)$ is defined.

A supervisory control function for \textbf{G} is a map 
\begin{equation}
	\textbf{V}:L(\textbf{G}) \rightarrow \Gamma
\end{equation}
in which\footnote{Operator $\text{Pwr}(\cdot)$ returns the power set of its argument set.}
\begin{equation}
	\Gamma = \{\gamma \in \text{Pwr}(\Sigma)|\gamma \supseteq \Sigma_{u}\}
\end{equation}
is the set of \textit{control patterns}. `\textbf{G} under supervision of \textbf{V}' is written as $\textbf{V}/\textbf{G}$. Given a sublanguage $M \subseteq L_{m}(\textbf{G})$ we define the \textit{marked behavior} of $\textbf{V}/\textbf{G}$ as 
\begin{equation}
	L_{m}(\textbf{V}/\textbf{G}) := L(\textbf{V}/\textbf{G}) \cap M
\end{equation}
$\textbf{V}$ is a 
\textit{marking nonblocking supervisory control} (MNSC) for the pair\footnote{Operator $\overline{(\cdot)}$ yields the prefix closure of its regular language argument.} $(M,\textbf{G})$ if
\begin{equation}
	\overline{L_{m}(\textbf{V}/\textbf{G})} = L(\textbf{V}/\textbf{G}).
\end{equation}
In practice, \textbf{V} is implemented by a \textit{supervisor} representing the maximally permissive controlled behavior $L_{m}(\textbf{V}/\textbf{G})$ subject to a generator specification, say \textbf{S}; we denote this computation by 
\begin{equation}
	\textbf{V} = \textbf{supcon} (\textbf{G},\textbf{S}).
\end{equation}
For details see, e.g., \cite{ramadge1987supervisory}. Given a generator \textbf{V}, and states $q,q' \in Q_{\textbf{V}}$, backtracking forcibility analysis is taken into account as
\begin{equation}
	Z := \mathcal{BFA} (\textbf{V},q,q')
\end{equation}
which yields set of all forcible paths which reach $q'$ from $q$ using backtracking forcibility \cite{macktoob2017auto}.
\section{Completeness-Seeking Supervision}
\label{sec:supervision}
Each positioner can rotate around its axis according to a specific number of discrete movements. In particular, given a positioner $\mathscr{P}$ and its step size $n^{\mathscr{P}}$, $\mathscr{P}$ rotates $2\pi/n^{\mathscr{P}}$ radians during each of its motion steps. From now on, $\mathcal{P}$ denotes the set of all positioners corresponding to a specific telescope. We define the notions of ``forward event'' and ``backward event'' as follows.
\begin{defn}[Forward/Backward Events]
	Let $\mathscr{P} \in \mathcal{P}$ be a positioner with motion step size $n^{\mathscr{P}}$. Then, if the controllable \textit{forward event} $v^{\mathscr{P}}$ (resp., the \textit{backward event} $w^{\mathscr{P}}$) is enabled, $\mathscr{P}$ rotates $2\pi/n^{\mathscr{P}}$ radians in a clockwise (resp., counterclockwise) direction around its axis. 
\end{defn}  
We also define counters corresponding to the number of the required movements to reach a target in different directions as follows.
\begin{defn}[Forward/Backward Counters]
	Let $\mathscr{P} \in \mathcal{P}$ be a positioner with motion step size $n^{\mathscr{P}}$. Given a specific position at the motion space of $\mathscr{P}$, the \textit{forward counter} $n^{\mathscr{P}}_v$ (resp., the \textit{backward counter} $n^{\mathscr{P}}_w$) represents the required number of clockwise (resp., counterclockwise) motion steps to reach the target position from the current position of $\mathscr{P}$.
\end{defn} 
\begin{figure}[tb]
	\centering
	\includegraphics[scale=0.5]{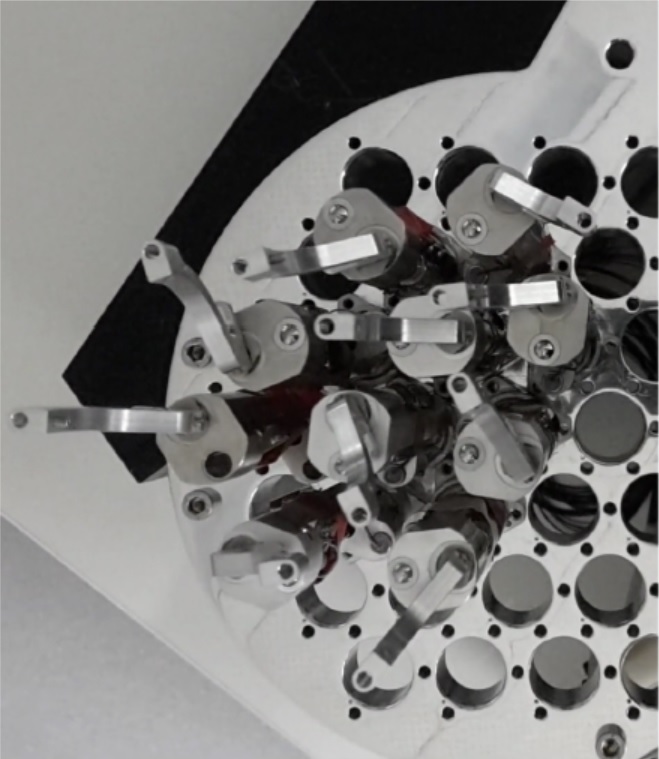}
	\caption{The hexagonal formation of a set of positioners (reprinted from \cite{horler2018robotic} with permission)}
	\label{fig:hex}
\end{figure}
In a hexagonal neighborhood, see, Figure \ref{fig:hex}, around a specific positioner, the maximum number of the positioners which may be involved in a colliding scenario to block each other is three. So, the relative priority of those typical positioners shall be defined in the language of events as follows.
\begin{defn}[Priority Events]
	Let $\mathcal{N} \subset \mathcal{P}$ be a set of two positioners around positioner $\mathscr{P} \in \mathcal{P}$ with which construct a colliding scenario. Let also $r(\mathscr{P})$ denote the relative priority of $\mathscr{P}$'s target compared to those of the positioners of $\mathcal{N}$. Then, we define the following controllable events.
	\begin{itemize}
		\item if $(\forall \mathscr{P}' \in \mathcal{N})~ r(\mathscr{P}) < r(\mathscr{P}')$, then the controllable event $l^{\mathscr{P}}$ is exclusively enabled;
		\item if $(\forall \mathscr{P}' \in \mathcal{N})~ r(\mathscr{P}) > r(\mathscr{P}')$, then the controllable event $h^{\mathscr{P}}$ is exclusively enabled;
		\item otherwise, the controllable event $m^{\mathscr{P}}$ is exclusively enabled.
	\end{itemize}
\end{defn}
\begin{figure}[tb]
	\centering
	\includegraphics[scale=0.6]{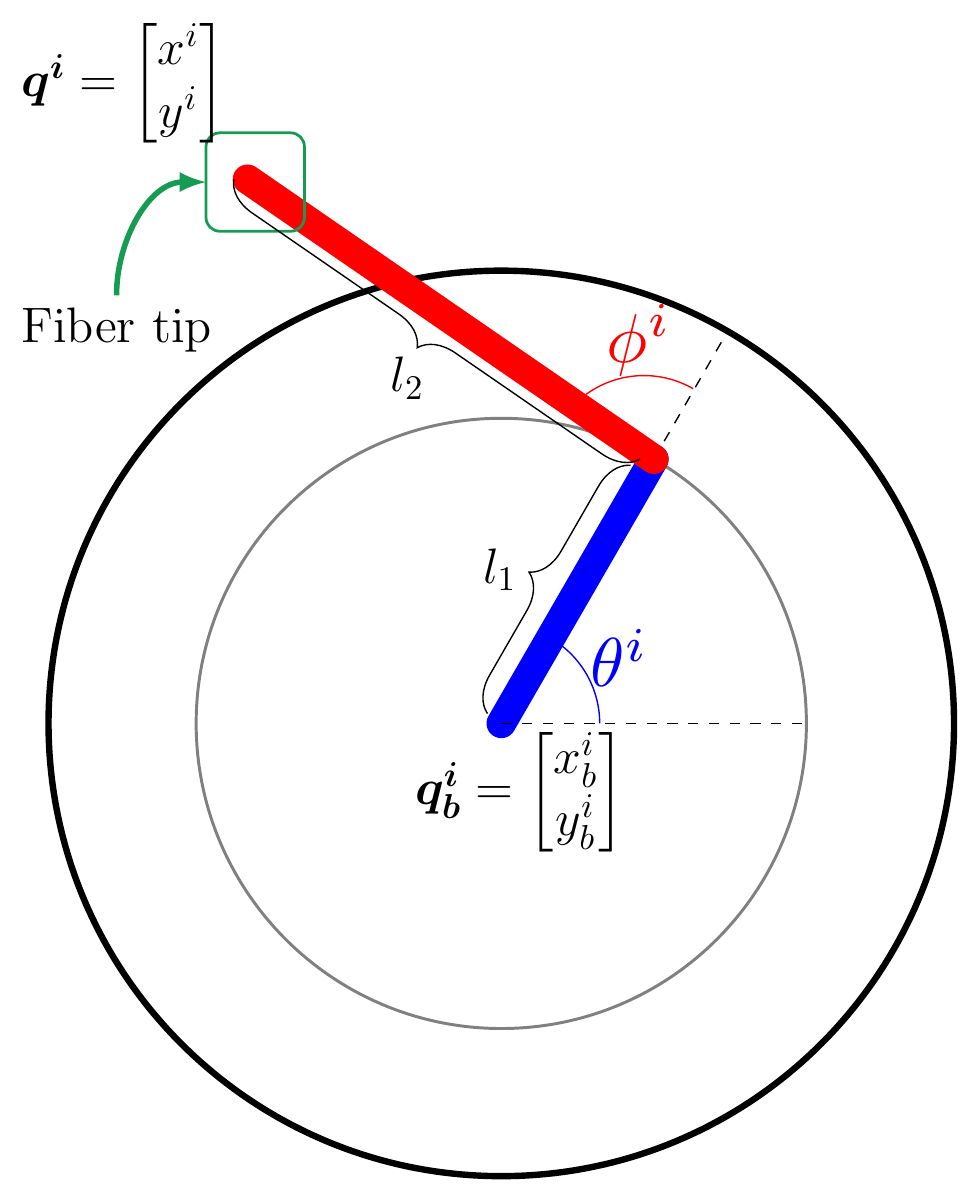}
	\caption{The discrete-event model of a typical positioner $\mathscr{P}$ with  with $n^{\mathscr{P}} = 3$ (The indices of the events are intentionally omitted for the better readability of the image.)}
	\label{fig:plant}
\end{figure}
We take some uncontrollable events into account to keep track of any potential colliding situation around a typical positioner as follows.
\begin{defn}[Collision/Free Events]
	Given positioner $\mathscr{P} \in \mathcal{P}$, if a neighboring positioner enters the the safety zone of $\mathscr{P}$, uncontrollable \textit{collision event} $k^{\mathscr{P}}$ occurs. In contrast, if all of the neighboring positioners which have resided at the safety zone of $\mathscr{P}$ leave that area, then uncontrollable \textit{free event} $e^{\mathscr{P}}$ occurs.
\end{defn}
We define the notion of a positioner as a discrete-event system as follows.
\begin{defn}[Positioner]
	Let $\mathscr{P} \in \mathcal{P}$ be a positioner. Given motion step size $n^{\mathscr{P}}$, forward event $v^{\mathscr{P}}$, backward event $w^{\mathscr{P}}$, forward counter $n_{v}^{\mathscr{P}}$, backward counter $n_{w}^{\mathscr{P}}$, priority events $l^{\mathscr{P}}$, $m^{\mathscr{P}}$, and $h^{\mathscr{P}}$, collision event $k^{\mathscr{P}}$, and free event $e^{\mathscr{P}}$, the generator $\textbf{G}^{\mathscr{P}}$ associated with $\mathscr{P}$ is defined as a discrete-event system
	\begin{equation}
	\textbf{G}^{\mathscr{P}} := (Q^{\mathscr{P}}_\textbf{G},\Sigma^{\mathscr{P}}_\textbf{G}, \delta^{\mathscr{P}}_\textbf{G},{q_{0}}^{\mathscr{P}}_\textbf{G},{Q_m}^{\mathscr{P}}_\textbf{G}),
	\end{equation}
	where
	\begin{itemize}[leftmargin=*]
		\item $Q^{\mathscr{P}}_\textbf{G} := \bigcup\limits_{j=0}^{{n^{\mathscr{P}}}}\{q_{j},q^{1}_{j},q^{2}_{j}\}$,
		\item $\Sigma^{\mathscr{P}}_\textbf{G} := \{v^{\mathscr{P}},w^{\mathscr{P}},l^{\mathscr{P}},m^{\mathscr{P}},h^{\mathscr{P}}, k^{\mathscr{P}}, e^{\mathscr{P}}\}$,
		\item \begin{align*}
			\delta^{\mathscr{P}}_\textbf{G} := (\forall j \vert 0 \le &j \le n^{\mathscr{P}}-1)\{ \delta(q_{j},v^{\mathscr{P}}):=q_{j+1},\\ &\delta(q_{j+1},w^{\mathscr{P}}):=q_{j},\\ &\delta(q^{j},k^{\mathscr{P}}):=q_{j}^{1},\\
			&\delta(q^{1}_{j},k^{\mathscr{P}}):=q^{2}_{j},\\
			&\delta(q^{2}_{j},m^{\mathscr{P}}):=q_{j}^{1},\\
			&\delta(q^{1}_{j},e^{\mathscr{P}}):=q_{j}\},\\
			&\delta(q^{2}_{j},h^{\mathscr{P}}) = \delta(q^{1}_{j},h^{\mathscr{P}}) := q_{j+1},\\
			&\delta(q^{2}_{j+1},l^{\mathscr{P}}) = \delta(q^{1}_{j+1},l^{\mathscr{P}}) := q_{j}\}\\ &\dot{\cup}\{\delta(q_{0},w^{\mathscr{P}}) := q_{n^{\mathscr{P}}},\\
			&\delta(q_{n^{\mathscr{P}}},v^{\mathscr{P}}) := q_{0}\},
		\end{align*}
		\item ${q_{0}}^{\mathscr{P}}_\textbf{G} := \text{a } q \in Q^{\mathscr{P}}_\textbf{G} \text{ according to the initial position of } \mathscr{P}$,
		\item ${Q_m}^{\mathscr{P}}_\textbf{G} := \text{a } q \in Q^{\mathscr{P}}_\textbf{G} \text{ based on the target position of } \mathscr{P}$.
	\end{itemize} 
\end{defn}
\begin{rem}
	We note that the occurrence of collision event $k^{\mathscr{P}}$ has to functionally preempt the other events eligible at a specific state of $\textbf{G}^{\mathscr{P}}$.
\end{rem}
Figure \ref{fig:plant} depicts the generator corresponding to the discrete-event model of a typical positioner $\mathscr{P}$ with $n^{\mathscr{P}}=3$. Each forward (resp., backward) event moves $\mathscr{P}$ $2\pi/3$ radians in a clockwise (resp., counterclockwise) direction. If no neighboring positioner enters the safety zone of $\mathscr{P}$, then it always remains in the scope of the lower-level states of its generator\footnote{From now on, symbols $[\bm{q}]$ and $\braket{\sigma}$ denote state $q$ and event $\sigma$, respectively.}, i.e., $[\bm{0}]$, $[\bm{1}]$, and $[\bm{2}]$. Otherwise, it may encounter one or two neighboring positioners at its safety zone, thereby occurring $k^{\mathscr{P}}$ events at $[\bm{0}]$ and/or $[\bm{1}]$. In the presence of only one colliding positioner, if it leaves the safety zone of $\mathscr{P}$, then the free event $e^{\mathscr{P}}$ occurs and $\mathscr{P}$ returns to its normal state at the lower level of $\textbf{G}^{\mathscr{P}}$. If the colliding positioner remains in the safety zone of $\mathscr{P}$, then the positioner with the higher priority moves clockwise to its target, and the other one with the lower priority moves in a counterclockwise direction. So, not only the collision is avoided, but the deadlock is also handled. In the case of three colliding positioners, there are three relative priorities. In particular, the positioners with the highest and the lowest priorities moves clockwise and counterclockwise, respectively. Furthermore, the positioner with the medium priority remains at its position. However, its colliding state is changed since at least one of the other colliding peers goes far from it.

\begin{figure}[t]
	\centering
	\includegraphics[scale=0.3]{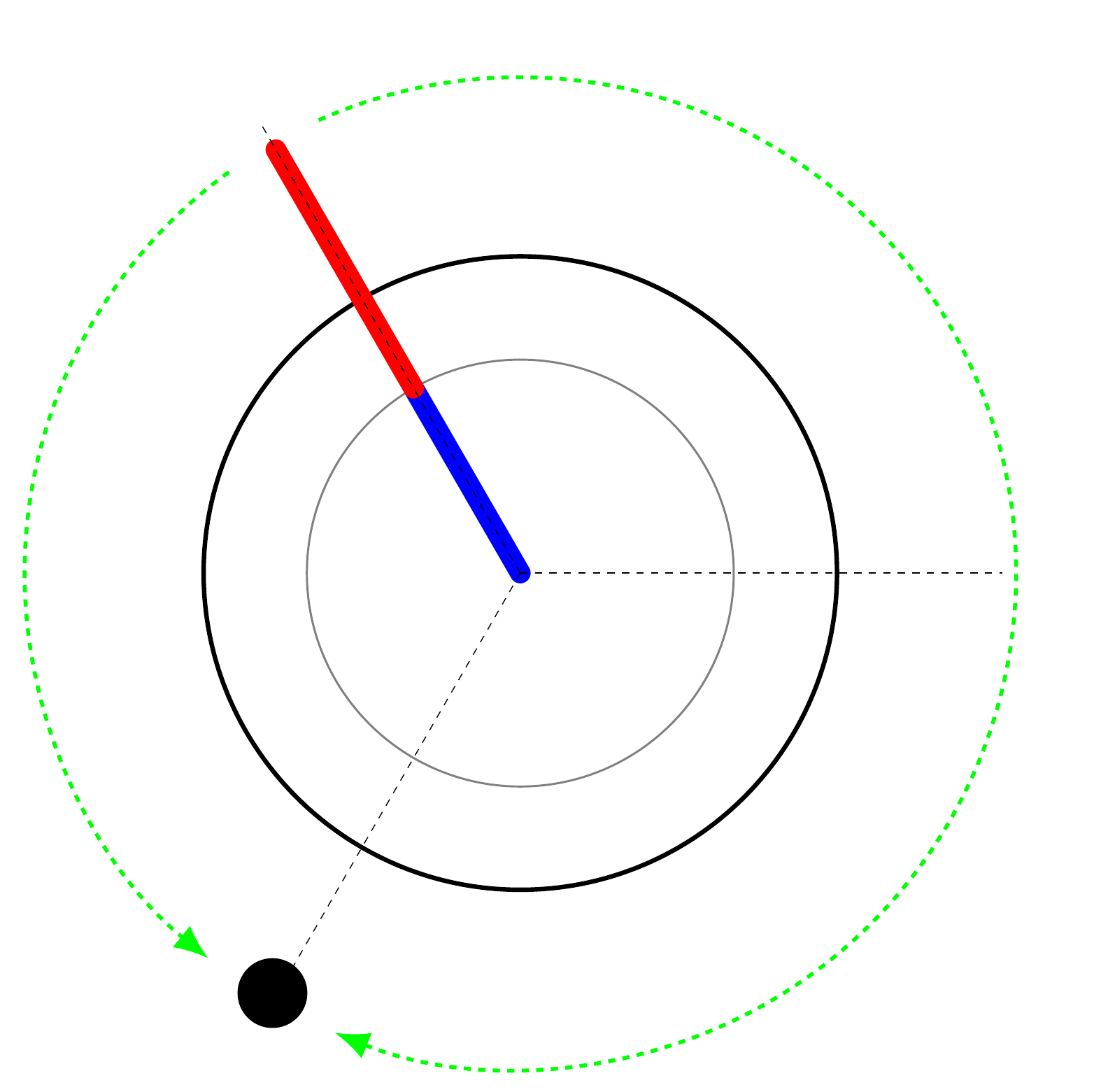}
	\caption{The initial and the target points corresponding to a typical positioner $\mathscr{P}$ with $n^{\mathscr{P}}=3$ (the blue and the red lines represent the arms of the depicted positioner.)}
	\label{fig:specExam}
\end{figure}

Now, we define the specification which determines the target position corresponding to a particular positioner.
\begin{defn}[Specification]
	Let $\mathscr{P} \in \mathcal{P}$ be a positioner. Then, considering forward counter $n_v^{\mathscr{P}}$ and backward counter $n_w^{\mathscr{P}}$ associated with $\mathscr{P}$, specification $\textbf{S}^{\mathscr{P}}$ corresponding to $\mathscr{P}$ is defined as a discrete-event system
	\begin{equation}
	\textbf{S}^{\mathscr{P}} := (Q^{\mathscr{P}}_\textbf{S},\Sigma^{\mathscr{P}}_\textbf{S}, \delta^{\mathscr{P}}_\textbf{S},{q_{0}}^{\mathscr{P}}_\textbf,{Q_m}^{\mathscr{P}}_\textbf{S}),
	\end{equation}
	where
	\begin{itemize}[leftmargin=*]
		\item $Q^{\mathscr{P}}_\textbf{S} := \biggl\{\bigcup\limits_{j=1}^{n_v^{\mathscr{P}}+n_w^{\mathscr{P}}-2}q_{j}\biggr\}\dot{\cup} \{q_{0},q_m\}$,
		\item $\Sigma^{\mathscr{P}}_\textbf{S} := \{v^{\mathscr{P}},w^{\mathscr{P}}\}$,
		\item \begin{align*}
			\delta^{\mathscr{P}}_\pi := (\forall j \vert 0 \le j \le n_v^{\mathscr{P}}-2)&\{ \delta(q_{j},v^{\mathscr{P}}):=q_{j+1},\\ &\delta(q_{j+1},w^{\mathscr{P}}):=q_{j}\} \dot{\cup}\\
			&\hspace*{-3cm}(\forall j \vert n_v^{\mathscr{P}}+1 \le j \le n_w^{\mathscr{P}}-1)\{ \delta(q_{j},w^{\mathscr{P}}):=q_{j+1},\\ &\hspace*{-1.5cm}\delta(q_{j+1},v^{\mathscr{P}}):=q_{j}\} \dot{\cup}\\
			&\hspace*{-1.5cm}\{\delta(q_{n_v^{\mathscr{P}}-1},v^{\mathscr{P}}) = \delta(q_{n_w^{\mathscr{P}}-2},w^{i}) := q_m\}
			\end{align*}
			\item ${q_{0}}^{\mathscr{P}}_\textbf{S} := q_{0}$,
			\item ${Q_m}^{\mathscr{P}}_\textbf{S} := q_m$.
	\end{itemize} 
\end{defn}
The specification corresponding to a particular positioner in fact determines the paths via which the positioner reaches its target position from its current position in both clockwise and counterclockwise directions. For example, suppose a positioner $\mathscr{P}$ with $n^{\mathscr{P}}=3$. The initial and the target positions of $\mathscr{P}$ are illustrated in Figure \ref{fig:specExam}, and Figure \ref{fig:spec} depicts its corresponding specification.
\begin{figure}[tb]
	\centering
	\includegraphics[scale=1]{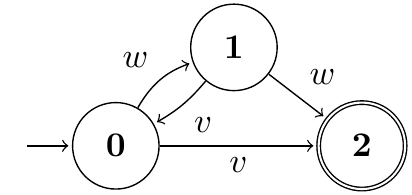}
	\caption{The specification corresponding to the Figure \ref{fig:specExam} (The event labels are intentionally omitted to improve the figure clarity.)}
	\label{fig:spec}
	\end{figure}

Now, we utilize supervisory control theory to generate a supervisor in which the solutions to the coordination problem are embedded. In particular, we first compute the overall model \textbf{G} of the positioners system by synchronizing\footnote{$n$-ary operator $\parallel$ computes the synchronous product of its argument generators.} the generators corresponding to all of the positioners.
\begin{equation}
\textbf{G} := \mathop{\parallel}_{\mathscr{P} \in \mathcal{P}} \textbf{G}^{\mathscr{P}} 
\end{equation}
We also obtain the overall model \textbf{S} of the specifications.
\begin{equation}
\textbf{S} := \mathop{\parallel}_{\mathscr{P} \in \mathcal{P}} \textbf{S}^{\mathscr{P}} 
\end{equation} 
Then, we\footnote{Operator $\textbf{allevents}(\cdot)$ returns a one-state generator with a selfloop including all of the states of its argument generator.}\footnote{Operator $\textbf{supcon}(\cdot,\cdot)$ computes a nonblocking supervisor corresponding to its first-argument generator with respect to its second-argument generator.} synthesize the coordination supervisor \textbf{V}.
\begin{equation}
\label{eq:V}
\textbf{V} = \textbf{supcon}(\textbf{G},(\textbf{allevents}(\textbf{G})\parallel\textbf{S}))
\end{equation}
The following theorem determines the role of a coordination supervisor in the definition and the solution to the complete coordination problem associated with it. 
\begin{thm}
	Let $\mathcal{P}$ be a set of positioners modeled by $\{\textbf{G}^{\mathscr{P}}|\mathscr{P} \in \mathcal{P} \}$ whose specifications are denoted by $\{\textbf{S}^{\mathscr{P}}|\mathscr{P} \in \mathcal{P} \}$. Let \textbf{V} be a coordination supervisor synthesized to control $\textbf{G} := \displaystyle\mathop{\parallel}_{\mathscr{P} \in \mathcal{P}} \textbf{G}^{\mathscr{P}}$ with respect to $\textbf{S} := \displaystyle\mathop{\parallel}_{\mathscr{P} \in \mathcal{P}} \textbf{S}^{\mathscr{P}}$. Then,
	\begin{enumerate}[label=(\roman*)]	
		\item \textbf{V} has only one marked state;
		\item \textbf{G} is completely coordinated at the marked state of \textbf{V}.
	\end{enumerate}
\end{thm}
\begin{proof}
	$(i)$ By the definition, each element of $\{\textbf{G}^{\mathscr{P}}|\mathscr{P} \in \mathcal{P} \}$ and $\{\textbf{S}^{\mathscr{P}}|\mathscr{P} \in \mathcal{P} \}$ has only one marked state. Then, according to the definition of the synchronous product operator \cite{wonham2017supervisory}, each of \textbf{G} and \textbf{S} has only one final state. Additionally, since \textbf{allevents(G)} includes only a single marked state, $\textbf{allevents}(\textbf{G})\parallel\textbf{S}$ has also only one marked state. Since every generator argument of \textbf{supcon} has only one marked state, we conclude that \textbf{V} possesses only one marked state, as well.\\ 
	$(ii)$ Each positioner is exclusively reached at its target position when its corresponding generator and the generator of its specification reside at their marked states. Then, the marked states of \textbf{G} and \textbf{S} represent the states whose simultaneous occupancies indicate the reachability of all of their corresponding positioners. If \textbf{V} reaches its marked state, then the cited simultaneous occupancies are realized. Therefore, \textbf{G} is complete when \textbf{V} resides at its marked state.
	\qed 
\end{proof}
\begin{algorithm}[b]
	\renewcommand{\thealgorithm}{}
	\caption{Complete Coordination Checker ($\mathcal{CCC}$)}
	\label{alg:CCC}
	\begin{algorithmic}[1]
		\INPUT 
		\Statex $\textbf{V}$ \Comment Coordination supervisor
		\OUTPUT 
		\Statex Complete coordination feasibility
		\algrule[1pt]
		\If{$\mathcal{BFA}(\textbf{V},{q_{0}}_\textbf{V},{q_{m}}_{\textbf{V}}) \neq \varnothing$}
		\Return True
		\Else \Return False 
		\EndIf
	\end{algorithmic}
\end{algorithm}
A supervisor always starts its supervision process from its initial state. So, the strings of events which can forcibly reach the marked state from the initial state are in fact the solutions to the complete coordination problem. Backtracking forcibility is a promising approach to find the forcible paths which reach one state from another \cite{macktoob2017auto,macktoobianautomatic}. Thus, the complete coordination problem can be written in the language of backtracking forcibility as follows.
\begin{defn}[Forcibility]
	Let \textbf{G} be a discrete-event system. Given two states $q,q' \in Q_{\textbf{G}}$ and a string $s \in \Sigma^{*}_{\textbf{G}}$, if $s$ forcibly reaches $q'$ from $q$, then the following ternary relation holds.
	\begin{equation}
	\mathcal{F}(s,q,q')
	\end{equation}
\end{defn}
\begin{plm}[Complete Coordination]
	Let \textbf{G} be a discrete-event system representing a system of positioners, \textbf{S} be the specification associated with \textbf{G}, and \textbf{V} be the coordination supervisor synthesized based them. Let ${q_{0}}_\textbf{V}$ and ${q_{m}}_{\textbf{V}}$ be the initial and the marked states of \textbf{V}, respectively. Subject to an appropriate specification of forcibility, check whether the following \textit{completeness condition} holds.
	\begin{equation}
	\label{eq:condition}
	(\exists s \in \Sigma^{*}_{\textbf{V}}) \mathcal{F}(s,{q_{0}}_\textbf{V},{q_{m}}_{\textbf{V}})
	\end{equation}
\end{plm}   
We solve the problem above using the backtracking forcibility notion of supervisory control theory as illustrated by the $\mathcal{CCC}$ algorithm. In particular, the completeness condition (\ref{eq:condition}) is checked. $\mathcal{BFA}$ function collects the forcible paths (if any) belonging to the marked behavior of \textbf{V} which forcibly reach the marked state of \textbf{V} from its initial state. If the result is nonempty, then the complete coordination of the desired system is feasible.
\begin{figure}[t]
	\begin{center}
		\includegraphics[scale=0.25]{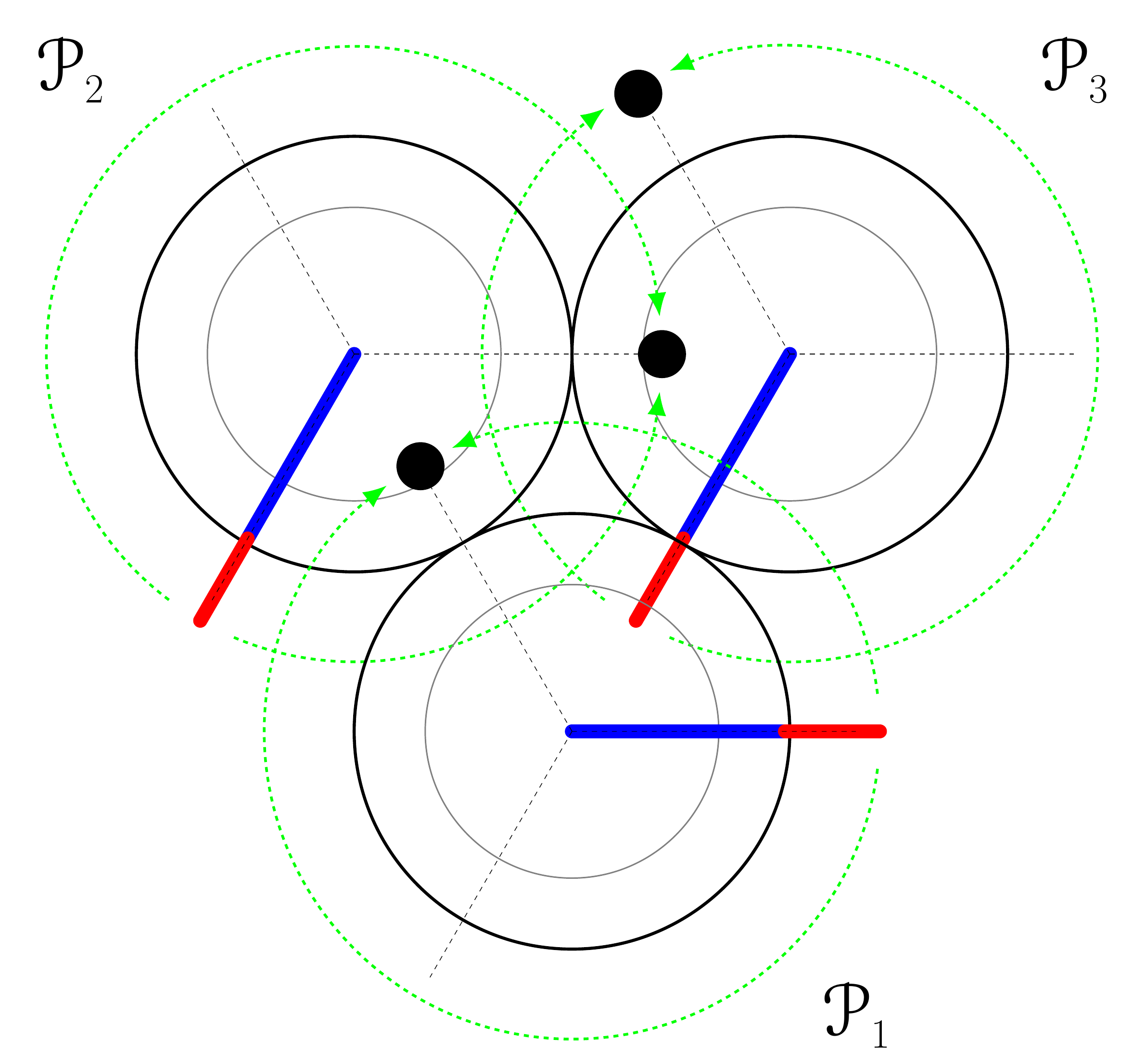}
		\caption{A coordination problem}
		\label{fig:examp}
	\end{center}
\end{figure} 

In the next section, we present an example whose completeness problem is assessed using $\mathcal{CCC}$ algorithm.
\section{Example}
\label{sec:example}
We solve a complete coordination problem using supervisory control theory implemented in TCT software \cite{tct2017supervisory}. We consider three positioners $\textbf{G}^{\mathscr{P}_{1}}$, $\textbf{G}^{\mathscr{P}_{2}}$, and $\textbf{G}^{\mathscr{P}_{3}}$ as depicted in Figure  \ref{fig:examp} whose initial states are $[\bm{0}]$, $[\bm{6}]$, and $[\bm{6}]$, respectively. The events of the positioners are specified in Table \ref{tbl:example} where events with odd labels are controllable, and those with even labels are uncontrollable.
\begin{table}[b]
	\begin{center}
		\caption{Event specifications of the example positioners (the indices of the events are intentionally omitted for better readibility.)}
		\label{tbl:example}
		\begin{tabular}{cccccccc}
			\toprule         
			Positioner&$v$&$w$&$l$&$m$&$h$&$k$&$e$\\
			\midrule 
			$\mathscr{P}_{1}$&71&73&11&13&15&10&40\\
			$\mathscr{P}_{2}$&81&83&21&23&25&20&50\\
			$\mathscr{P}_{3}$&91&93&31&33&35&30&60\\
			\bottomrule
		\end{tabular}
	\end{center}
\end{table}
Suppose the relative priorities of the positioners are the same as their indices. Specifications $\textbf{S}^{\mathscr{P}_{1}}$, $\textbf{S}^{\mathscr{P}_{2}}$, and $\textbf{S}^{\mathscr{P}_{3}}$ are illustrated in Figures \ref{fig:s1}, \ref{fig:s2}, and \ref{fig:s3}, respectively.

We compute the coordination supervisor \textbf{V} according to (\ref{eq:V}) whose marked state is $[\bm{14}]$. Using the $\mathcal{CCC}$ algorithm, we obtain the (shortest) forcible path\footnote{Symbol $\braket{\sigma_{1}, \cdots,\sigma_{n}}$ denotes a string constructed by its argument events.} $Z = \braket{73,10,30,15,31,81,20}$ from $[\bm{0}]$ to $[\bm{14}]$. In particular, \textbf{V} commands $\mathscr{P}_1$ to rotate in a counterclockwise direction by enabling $\braket{73}$. Then, $\mathscr{P}_1$ enters the safety zone of $\mathscr{P}_3$. So, the collision events $\braket{10}$ and $\braket{30}$ occur. Since $\mathscr{P}_1$ has a higher priority, $\braket{15}$ drives $\mathscr{P}_1$ to its target point in a clockwise direction. In contrast, $\mathscr{P}_3$ goes far from $\mathscr{P}_1$ to avoid any collision by $\braket{31}$. So, $\mathscr{P}_3$ eventually reaches its target point. $\braket{71}$ initially moves $\mathscr{P}_2$ in a clockwise direction. However, it enters the safety zone of $\mathscr{P}_1$, so $\braket{20}$ occurs, and $\mathscr{P}_2$ changes it direction which gives rise to its convergence.
\begin{figure}[bt]
	\centering
	\includegraphics[scale=1]{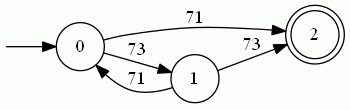}
	\caption{The generator of $\textbf{S}^{\mathscr{P}_{1}}$}
	\label{fig:s1}
\end{figure}
\begin{figure}[tb]
	\centering
	\includegraphics[scale=1]{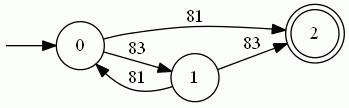}
	\caption{The generator of $\textbf{S}^{\mathscr{P}_{2}}$}
	\label{fig:s2}
\end{figure}
\begin{figure}[tb]
	\centering
	\includegraphics[scale=1]{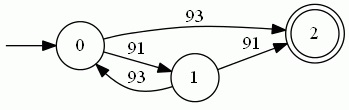}
	\caption{The generator of $\textbf{S}^{\mathscr{P}_{3}}$}
	\label{fig:s3}
\end{figure} 
\section{Conclusions}
\label{sec:conc}	
We reported a supervisory-control-based solution to the complete coordination problem of robotic fiber positioners. We illustrated how to model positioners systems and their specifications as discrete-event systems. We then generated a coordination supervisor in which the solution(s) to the problem were embedded. We found the solution(s) using the backtracking forcibility. 

One observes that our proposed method is straightforward to check the complete coordination feasibility for 1 DoF positioners. However, the extensive discretization of the space motions of positioners may give rise to intractably large discrete-event systems. A future stream of research would be the usage of state tree structures to efficiently coordinate extremely complex systems of positioners. Furthermore, untimed discrete-event systems do not consider any temporal requirements or constraints in the course of modeling and control process. Thus, one may take timed discrete-event systems into account to model positioners not to skip potential important temporal characteristics of the system.   

\nocite{*}
\bibliographystyle{IEEEtran}
\bibliography{ifacconf}
            
\end{document}